\documentclass[conference,a4paper]{IEEEtran}
\pdfoutput=1

\usepackage{mathtools}
\usepackage{amsmath}
\usepackage{amssymb}
\usepackage{amsfonts}
\usepackage{amsthm}
\usepackage{refcheck}
\usepackage{graphicx}
\usepackage{graphics}
\usepackage{latexsym}
\usepackage[dvips]{epsfig}
\usepackage[nospace]{cite}
\usepackage{subfigure}
\usepackage{setspace}
\usepackage{tabularx}
\usepackage{epstopdf}
\usepackage{tikz,pgf}
\usepackage{texdraw}
\usepackage{empheq}
\usetikzlibrary{arrows}

\newtheorem{Definition}{Definition}
\newtheorem{Lemma}{Lemma}

\newtheorem{Theorem}{Theorem}

\def\Pr{{\rm \bold {Pr}}}
\def\E{{\rm \bold  E}}

\def\ci{\perp\!\!\!\perp}

\allowdisplaybreaks

\begin{document}

\sloppy

%% Paper Title
%% You can use linebreaks \\ within to get better formatting as
%% desired. 
\title{Two-Hop Interference Channels: \\ Impact of Linear Time-Varying Schemes} 

% \author{
%   \IEEEauthorblockN{Hui-Ting Chang and Stefan M.~Moser}
%   \IEEEauthorblockA{Department of Electrical and Computer Engineering\\
%     National Chiao Tung University (NCTU)\\
%     Hsinchu, Taiwan\\
%     Email: \{email-of-hui-ting,email-of-stefan\}@ieee.org} 
% }

 \author{
   \IEEEauthorblockN{Ibrahim Issa, Silas L. Fong, and A. Salman Avestimehr}
   \IEEEauthorblockA{School of Electrical and Computer Engineering\\
     Cornell University, Ithaca, New York, USA\\
     Emails: ii47@cornell.edu, lf338@cornell.edu, avestimehr@ece.cornell.edu} 
 }

%\author{
%  \IEEEauthorblockN{Ibrahim Issa}
%  \IEEEauthorblockA{Cornell University \\
%    Email: ii47@cornell.edu} 
%  \and
%  \IEEEauthorblockN{Silas L. Fong}
%  \IEEEauthorblockA{Cornell University \\
%    Email: silas\_fong@hotmail.com}
%  \and
%  \IEEEauthorblockN{A. Salman Avestimehr}
%  \IEEEauthorblockA{Cornell University \\
%    Email: avestimehr@ece.cornell.edu}
%}

%% To balance the two columns, you should reduce the text-height of
%% the last page using the following command:
%%%%%%%%%%%%%%%%%%%%%%%%%%%%%%%%%%%%%%%%%%%%%%%%%%%%%%%%%%%%%%%%%%%%%
%\addtolength{\textheight}{-9.35cm}
%%%%%%%%%%%%%%%%%%%%%%%%%%%%%%%%%%%%%%%%%%%%%%%%%%%%%%%%%%%%%%%%%%%%%
%% with an appropriate value. This command must be place on the second
%% last page, i.e., for a one-page abstract here, for a two-page
%% abstract right after the \maketitle command.

%% Create the title:
\maketitle

%% Abstract: 
%% For the final version of the accepted paper, please make sure you
%% remove the comment "THIS PAPER IS ELIGIBLE FOR THE STUDENT PAPER
%% AWARD."
%%
\begin{abstract}
We consider the two-hop interference channel (IC) with constant real channel coefficients, which consists of two source-destination pairs, separated by two relays. We analyze the achievable degrees of freedom (DoF) of such network when relays are restricted to perform scalar amplify-forward (AF) operations, with possibly time-varying coefficients. We show that, somewhat surprisingly, by providing the flexibility of choosing time-varying AF coefficients at the relays, it is possible to achieve 4/3 sum-DoF. We also develop a novel outer bound that matches our achievability, hence characterizing the sum-DoF of two-hop interference channels with time-varying AF relaying strategies.
\end{abstract}

\section{Introduction}

Multi-hopping is typically viewed as an effective approach to extend the coverage range of wireless networks, by bridging the gap between the sources and destinations via relays. However, it has also the potential to significantly impact network capacity by enabling new interference management techniques (see, e.g., \cite{Mohajer,Simeone,Thejaswi}). In particular, from the degrees of freedom (DoF) perspective that is the focus of this paper, authors in~\cite{Jafar} considered a two-hop complex interference channel (IC) consisting of two sources, two relays, and two destinations, and they showed by introducing a new scheme called aligned-interference-neutralization that the sum-DoF of this network is 2 (i.e., twice the sum-DoF of a single-hop IC). More recently, authors in~\cite{IlanKKK} have considered two-hop interference networks with $K$ sources, $K$ relays, and $K$ destinations, and they showed by developing a new scheme named aligned-network-diagonalization that relays have the potential to asymptotically cancel the interference between all source-destination pairs, hence the cut-set bound is achievable (i.e., sum-DoF of $K$).

While the aforementioned results essentially demonstrate that significant DoF gains can be achieved by carefully designing the interference management strategies in multi-hop interference networks, they often require complicated relaying strategies (such as, utilizing rational dimensions for neutralizing the interference when the channels are not time-varying). In this paper, we take a complementary approach and ask how much of these DoF gains can be realized if we limit the operation of relays to simple \emph{scalar linear} strategies?

We focus on two-hop interference channels with constant \emph{real} channel coefficients (i.e., slow fading), and assume that the relays are allowed to perform only scalar amplify-forward (AF) operations with possibly \emph{time-varying} AF coefficients. It is easy to see that if AF coefficients of the relays remain constant during the course of the scheme, then the problem will induce to a single-hop IC, in which the sum-DoF is at most $1$. However, we show that, somewhat surprisingly, by providing the flexibility of choosing time-varying AF coefficients at the relays, a sum-DoF of $4/3$ is achievable.

The key idea behind the achievability strategy is that the flexibility of choosing the relay AF factors allows canceling, in any specific time slot, one source signal from one destination. So, we use this flexibility to guarantee that, for each destination, at most one third of its received symbols are distinct interference symbols, which allows it to achieve $2/3$ DoF.

To derive the outer bound, we break the end-to-end mutual information achieved by any scheme into five different groups, based on five distinct states that scalar linear schemes can create at each time-step. We then proceed to prove three outer bounds that effectively capture the tension between these groups. Analyzing the three bounds yields that the sum-DoF is upper bounded by $4/3$ almost surely.

\section{Problem Setting \& Main Result} \label{setting}

As illustrated in Figure \ref{channel}, we consider the two-hop IC, consisting of two sources, two relays, and two destinations.\\  
\begin{figure}[h]
\centering 
{\includegraphics[scale=1.2]{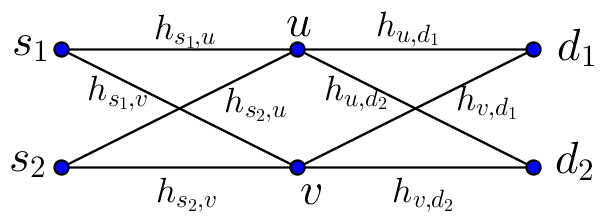}}
%\vspace{-3mm}
\caption{Two-hop IC.}\label{channel}
\end{figure}
%\vspace{1mm}

We denote the two sources by $s_1$ and $s_2$, the two relays by $u$ and $v$, and the destinations by $d_1$ and $d_2$.
Each source $s_i$ has a message $W_i$ intended for $d_i$ ($i \in \{1,2\}$), and $W_1 \ci W_2$. \\ 

Let $\mathbf{H_1}=\begin{bmatrix} h_{s_1,u} & h_{s_2,u} \\ h_{s_1,v} & h_{s_2,v} \end{bmatrix}$ and $\mathbf{H_2}=\begin{bmatrix} h_{u,d_1} & h_{v,d_1} \\ h_{u,d_2} & h_{v,d_2} \end{bmatrix}$ be the channels of the first and second hop, respectively. We assume that the channel gains are real-valued and drawn from a continuous distribution,  fixed during the course of communication, and known at all nodes.

%We denote the transmit signals of 

The transmit signal of $s_i$ and relay $r$ at time $k$ are respectively denoted by $X_{i,k} \in \mathbb{R}$ and $X_{r,k} \in \mathbb{R}$, $i \in \{1,2\}$ and $r \in \{u,v\}$. The received signal of relay $r$ at time $k$ is
\begin{equation*} \label{firstChannelAddition}
Y_{r,k}=h_{s_1,r}X_{1,k}+h_{s_2,r}X_{2,k}+Z_{r,k}, \text{\qquad $r \in \{u,v\}$, $k \in \mathbb{N},$}
\end{equation*}
and for destination $d_i$, the received signal at time $k$ is
\begin{equation*}
Y_{i,k}=h_{u,d_i}X_{u,k}+h_{v,d_i}X_{v,k}+Z_{d_i,k},  \text{\qquad $i \in \{1,2\}$, $k \in \mathbb{N},$}
\end{equation*}
where $Z_{r,k}$'s and $Z_{d_i,k}$'s are i.i.d (over time and with respect to each other) noise terms distributed as $\sim\mathcal{N}(0,1)$, which are also independent of the messages $\{W_1,W_2\}$.
We will use $X^n$ to denote a random column vector $[X_1\  X_2\  \ldots\  X_n]^T$. Also, for any $\mathcal{S}\subseteq \{1, 2, \ldots, n\}$, we let $X^\mathcal{S}$ denote $\{X_k | k\in\mathcal{S}\}$.

\begin{Definition} \label{defCode}
An $(n,R_1,R_2)$-scheme with power constraint $P$ on the two-hop IC consists of the following:
\begin{enumerate}
\item A message set $\mathcal{W}_i=\{1,2,\dots,2^{nR_i}\}$ at $s_i$, $i \in \{1,2\}$.
\item An encoding function $f_i$: $\mathcal{W}_i \rightarrow \mathcal{X}_i^n$ for each source $s_i$, $i \in \{1,2\}$, such that $X_i^n=f_i (W_i)$, and every codeword $x_i^n$ satisfies the power constraint $ \sum_{k=1}^{n}{x^2_{i,k}} \leq nP $.
\item A relaying function $f_{r,k}$: $\mathcal{Y}_r^{k-1} \rightarrow \mathcal{X}_r$ at $r$ for each $r \in \{u,v\}$ and each $k \in \{1,2,\dots,n\}$, such that $X_{r,k}=f_{r,k}(Y_r^{k-1})$. In addition, every codeword $x_r^n$ should satisfy the power constraint $\sum_{k=1}^{n}{x^2_{r,k}} \leq nP$.
\item A decoding function $g_i$: $\mathcal{Y}_i^n \rightarrow \mathcal{W}_i$ for destination $d_i$, $i \in \{1,2\}$, such that $\hat{W}_i = g_i(Y_i^n)$.
\item The error probability $P_e^n$ of the scheme is defined as
$
P_e^n = \Pr\left( \bigcup_{i=1}^2 \{W_i \neq \hat{W}_i \} \right),
$
where each $W_i$ is chosen independently and uniformly at random from $\{1,2,\dots,2^{nR_i}\}$, $i \in \{1,2\}.$
\end{enumerate}
\end{Definition}

\begin{Definition} \label{defTVAF} {\bf{(Time-varying AF scheme)}}
Let $\mathcal{U}$ and $\mathcal{V}$ be two finite subsets of $\mathbb{R}$. An $(n,R_1,R_2)$-scheme on the two-hop IC is called a time-varying AF on $(\mathcal{U,V})$ if there exist ${\{\mu_k \in \mathcal{U}\}}_{k=1}^n$ and ${\{\lambda_k \in \mathcal{V}\}}_{k=1}^n$ such that, for each $k \in \{1,2,\dots,n\}$, $f_{u,k}(Y_u^{k-1})=\mu_k Y_{u,k-1}$ and $f_{v,k}(Y_v^{k-1})=\lambda_k Y_{v,k-1}$.
\end{Definition}

\begin{Definition} \label{defAchievable}
A rate pair $(R_1, R_2)$ is \textit{time-varying-AF-achievable on $(\mathcal{U},\mathcal{V})$} if there exists a sequence of $(n, R_1, R_2)$-schemes that are time-varying AF on $(\mathcal{U},\mathcal{V})$, s.t.  $\lim\limits_{n\rightarrow \infty} P_e^n= 0$.
\end{Definition}

\begin{Definition} \label{defDoF}
The sum-DoF achievable by time-varying AF, denoted by $\mathcal{D}$, is defined by
\[
\mathcal{D} = \sup_{\mathcal{U},\mathcal{V}}\lim_{P\rightarrow\infty}\! \sup\left\{\left. \frac{R_1+R_2}{\frac{1}{2}\log_2 P} \:\right|\parbox[c]{1.5 in}{$(R_1, R_2) \text{ is time-varying-}\\\text{AF-achievable on }(\mathcal{U},\mathcal{V})$}\! \right\}.
\]
\end{Definition}

\noindent The main result of the paper is the following theorem.
\begin{Theorem} \label{MainTh}
The sum-DoF of two-hop IC with time-varying AF schemes is $4/3$ for almost all values of channel gains.
\end{Theorem}
\noindent In particular, the channel gain conditions needed for Theorem \ref{MainTh} to yield 4/3 sum-DoF are as follows:
%Particularly, the channel conditions needed for Theorem  are blublublu:
\begin{align}
\nonumber & \text{(c-1)  All channel gains are non-zero.} \\
\nonumber &  \text{(c-2) rank}(\mathbf{H_i})=2, ~ i \in \{1,2\}.\\
\nonumber &  \text{(c-3) rank}\left(\mathbf{H}^i\stackrel{\Delta}{=}\begin{bmatrix} h_{u,d_1} h_{s_i,u} & h_{v,d_1} h_{s_i,v} \\ h_{u,d_2} h_{s_{\bar{i}},u} & h_{v,d_2} h_{s_{\bar{i}},v} \end{bmatrix}\right)=2, \\
\label{eq:conditions}& \qquad \ i \in \{1,2\}, \bar{i}=3-i.
\end{align}
It is easy to see that almost all values of channel gains satisfy the above conditions. In the rest of the paper, in which we prove Theorem~\ref{MainTh}, we assume that  conditions  (c-1)--(c-3) hold.

%\begin{itemize}
%\item[] (c-1)  All channel coefficients are non-zero. 
%\item[] (c-2) rank$(\mathbf{H_i})=2$, $i \in \{1,2\}$. \vspace{0.03in}
%\item[] (c-3) rank$\left(\mathbf{H}^i\stackrel{\Delta}{=}\begin{bmatrix} h_{u,d_1} h_{s_i,u} & h_{v,d_1} h_{s_i,v} \\ h_{u,d_2} h_{s_{\bar{i}},u} & h_{v,d_2} h_{s_{\bar{i}},v} \end{bmatrix}\right)=2$, \vspace{0.03in}
%\item[] \qquad $\ i \in \{1,2\}$, $\bar{i}=3-i$.
%\end{itemize}

\section{Achieving $4/3$ Sum-DoF by Time-Varying AF} \label{achieve}

The achievability scheme consists of three phases, during which each source sends two distinct symbols, and at the end of the three phases each receiver is able to reconstruct an interference free, but noisy, version of its desired symbols.

First note that, for time-varying AF strategies, the received signals at the destinations at each time $k$ can be written as 
\begin{equation}
\begin{split}
\begin{bmatrix} Y_{1,k} \\ Y_{2,k} \end{bmatrix} &  = \mathbf{H_2} \begin{bmatrix} \mu_k & 0 \\ 0 & \lambda_k \end{bmatrix} \mathbf{H_1} \begin{bmatrix} X_{1,k-1} \\ X_{2,k-1} \end{bmatrix} + \begin{bmatrix} \tilde{Z}_{1,k} \\ \tilde{Z}_{2,k} \end{bmatrix} \\
& = \mathbf{G}_k \begin{bmatrix} X_{1,k-1} \\ X_{2,k-1} \end{bmatrix} + \begin{bmatrix} \tilde{Z}_{1,k} \\ \tilde{Z}_{2,k} \end{bmatrix},
\end{split}
\end{equation}
where $\mu_k$ and $\lambda_k$ are the AF coefficients at time $k$, $\tilde{Z}_{i,k} = h_{u,d_i}\mu_k Z_{u,k-1}+h_{v,d_i}\lambda_k Z_{v,k-1}+Z_{d_i,k}$ is the effective noise at destination $d_i$, $i \in \{1,2\}$, and $\mathbf{G}_k =  \mathbf{H_2} \begin{bmatrix} \mu_k & 0 \\ 0 & \lambda_k \end{bmatrix} \mathbf{H_1}$ is the equivalent end-to-end channel matrix given by 
\begin{align} \label{Gk}
&\mathbf{G}_k = \\
&\begin{bmatrix} \mu_k h_{u,d_1} h_{s_1,u} \hspace{-1mm}+\hspace{-1mm} \lambda_k h_{v,d_1} h_{s_1,v} & \mu_k h_{u,d_1} h_{s_2,u} \hspace{-1mm}+\hspace{-1mm} \lambda_k h_{v,d_1} h_{s_2,v} \\  \mu_k h_{u,d_2} h_{s_1,u} \hspace{-1mm}+\hspace{-1mm} \lambda_k h_{v,d_2} h_{s_1,v} & \mu_k h_{u,d_2} h_{s_2,u} \hspace{-1mm}+\hspace{-1mm} \lambda_k h_{v,d_2} h_{s_2,v} \end{bmatrix}. \notag
\end{align}

\noindent For notational convenience, let $\mathbf{G}_k = \begin{bmatrix} \alpha_{1,k} & \beta_{1,k} \\ \alpha_{2,k} & \beta_{2,k} \end{bmatrix}$. Also, we will only need $\tilde{Z}_{i,k}$  for our analysis; so we will drop the tilde and write $Z_{i,k}$. Then, the received signal at destination $d_i$, $i \in \{1,2\}$, at time $k$ is
%Then the received signal at destination $d_i$, $i \in \{1,2\}$, at time $k$  can be written as
\begin{equation} \label{RSignal}
Y_{i,k}=\alpha_{i,k}X_{1,k}+\beta_{i,k}X_{2,k}+Z_{i,k},  \text{\quad $k \in \{1,2,\dots,n\}$}.
\end{equation}
\noindent Note that the variance of $Z_{i,k}$ depends only on channel coefficients and amplifying factors (chosen from ($\mathcal{U,V}$)), therefore it does not scale with $P$.

We will now describe the three phases of our time-varying AF achievability scheme in detail. Set $\mathcal{U}=\{c\}$, and $\mathcal{V}=\{0,-c h_{u,d_1} h_{s_2,u}/h_{v,d_1} h_{s_2,v}, -c h_{u,d_2} h_{s_1,u}/h_{v,d_2} h_{s_1,v}\}$, where the constant $c \in \mathbb{R}$ is chosen to satisfy the power constraint $P$ at the relays. More specifically, \\
\noindent $c= \min\left\{\sqrt{1/(h_{s_1,u}^2 \hspace{-1mm} +h_{s_2,u}^2 \hspace{-1mm}+1)}, l\sqrt{1/(h_{s_1,v}^2+h_{s_2,v}^2+1)} \right\}$,
where \\$l  =  \min\{ {|h_{v,d_1} h_{s_2,v}/h_{u,d_1} h_{s_2,u}| , |h_{v,d_2} h_{s_1,v}/h_{u,d_2} h_{s_1,u}|} \}.$ Note that the denominators are non-zero by condition (c-1).\\

\noindent \textbf{Phase 1.} In this phase, $s_1$ and $s_2$ send two symbols $a_1$ and $b_1$ respectively $(a_1^2,b_1^2 \leq P)$. We choose the AF factors at the relays  such that the interference from $s_2$ is canceled at $d_1$. More specifically, we set $\mu_1=c$ and $\lambda_1=-c h_{u,d_1} h_{s_2,u}/h_{v,d_1} h_{s_2,v}$. By inserting this choice of $\lambda_1$ and $\mu_1$ in (\ref{RSignal}), $d_1$ and $d_2$ will respectively receive
\begin{equation} \label{timeslot1}
y_{1,1}=\alpha_{1,1}a_1+z_{1,1}, \text{ and } y_{2,1}= \underbrace{\alpha_{2,1}a_1+\beta_{2,1}b_1}_{L_1(a_1,b_1)}+z_{2,1},
\end{equation}
where $\alpha_{1,1}\ne 0$ and $\beta_{2,1}\ne 0$ (due to conditions (c-1), (c-2), and (c-3) in (\ref{eq:conditions})), and $L_1(a_1,b_1)$ indicates a linear equation in $a_1$ and $b_1$. Thus, as shown in Figure \ref{scheme}(a), $d_1$ and $d_2$ now respectively have noisy versions of $a_1$ and $L_1(a_1,b_1)$.\\
%\vspace{-.1in}
\begin{figure}[!hbt]
\centering 
\vspace{-3mm}
\subfigure[tight][Phase 1]{\includegraphics[scale=0.78]{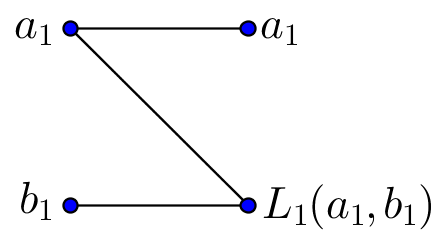}} %\hspace{0.1in}
\subfigure[tight][Phase 2]{\includegraphics[scale=0.78]{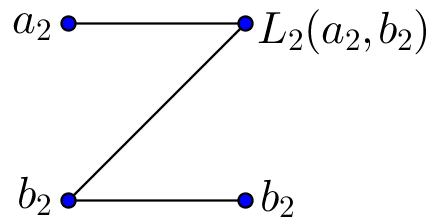}} %\hspace{0.1in}
\subfigure[tight][Phase 3]{\includegraphics[scale=0.78]{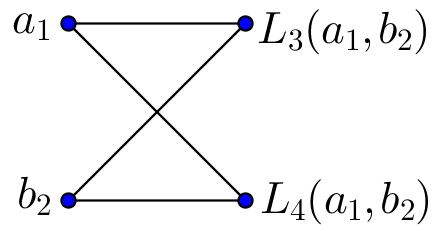}}
%\vspace{-2mm}
\caption{Illustration of achievability scheme. At each phase, the transmitted symbols are shown on the left. The  received signals at destinations are given on the right, where the noise is dropped and $L(x,y)$ denotes a linear combination of $x$ and $y$. } 
\label{scheme}
\end{figure}

\noindent \textbf{Phase 2.} In this phase, $s_1$ and $s_2$ send two new symbols $a_2$ and $b_2$ $(a_2^2,b_2^2 \leq P)$. However, this time, we cancel the effect of $s_1$ at $d_2$, by letting $\mu_2 = c $ and $\lambda_2 =-c h_{u,d_2} h_{s_1,u}/h_{v,d_2} h_{s_1,v}$. Then $d_1$ and $d_2$ will respectively receive
\begin{equation} \label{timeslot2}
y_{1,2}=\underbrace{\alpha_{1,2}a_2+\beta_{1,2}b_2}_{L_2(a_2,b_2)}+z_{1,2}, \text{ and } y_{2,1}=\beta_{2,2}b_2+z_{2,2},
\end{equation}
where $\alpha_{1,2} \ne 0$ and $\beta_{2,2} \ne 0$ (due to conditions (c-1), (c-2), and (c-3) in (\ref{eq:conditions})), and $L_2(a_2,b_2)$ indicates a linear equation in $a_2$ and $b_2$. Thus, as shown in Figure \ref{scheme}(b), $d_1$ and $d_2$ now respectively have noisy versions of $L_2(a_2,b_2)$ and $b_2$.

\vspace{2mm}
\noindent \textbf{Phase 3.} Now notice that, if, at phase 3, destination $d_1$ receives a linear combination of $a_1$ and $b_2$ ($L_3(a_1,b_2)$), then it can solve for (a noisy version of) $a_2$ given equations (\ref{timeslot1}) and (\ref{timeslot2}). Similarly, if $d_2$ receives $L_4(a_1,b_2)$ then it can also solve for (a noisy version of) $b_1$ given equations (\ref{timeslot1}) and (\ref{timeslot2}). Thus, as shown in Figure \ref{scheme}(c), in phase 3, $s_1$ sends $a_1$, $s_2$ sends $b_2$, and we choose $\mu_3=c$, and $\lambda_3=0$, so that $d_1$ and $d_2$ receive
{\small \begin{align} \label{timeslot3}
y_{1,3}=\underbrace{\alpha_{1,3}a_1+\beta_{1,3}b_2}_{L_3(a_1,b_2)}+z_{1,3}, ~
y_{2,3}=\underbrace{\alpha_{2,3}a_1+\beta_{2,3}b_2}_{L_4(a_1,b_2)}+z_{2,3},
\end{align}}where $\beta_{1,3} \ne 0$, and $\alpha_{2,3} \ne 0$ (due to condition (c-1) in (\ref{eq:conditions})). Therefore, after the three phases, $d_1$ can construct
  \begin{equation}
 y_1^{a_1} = a_1+z_{1,1}/\alpha_{1,1}, \text{\qquad\quad and }\label{firstDataStream}
   \end{equation}
   \begin{equation}
 y_1^{a_2} = a_2+\frac{1}{\alpha_{1,2}}z_{1,2}-\frac{\beta_{1,2}}{\alpha_{1,2}\beta_{1,3}}z_{1,3}+\frac{\alpha_{1,3}\beta_{1,2}}{\alpha_{1,1}\alpha_{1,2}\beta_{1,3}}z_{1,1}.
   \label{secondDataStream}
    \end{equation}
    from $(y_{1,1}, y_{1,2},y_{1,3})$. Let $\sigma_1^2$ and $\sigma_2^2$ be the variances of the noise terms in equations (\ref{firstDataStream}) and (\ref{secondDataStream}). Note that they depend only on channel coefficients and AF factors. Hence, they are constants that do not scale with $P$. Then, by using a proper outercode, we can achieve a rate of 
\begin{equation*}
\begin{split}
R_1  & = \frac{1}{6} \left( \log \left(1+\frac{P}{\sigma_1^2}\right)+\log \left(1+\frac{P}{\sigma_2^2} \right)  \right)   \geq \frac{1}{3} \log \frac{P}{\sigma_1 \sigma_2}. 
%\frac{1}{3} \log (\sigma_1^2 \sigma_2^2 ).
\end{split}
\end{equation*}
So $d_1$ can achieve $2/3$ DoF. Similarly, $d_2$ can also achieve $2/3$ DoF, hence achieving a total of $4/3$ sum-DoF. Note that a similar achievability scheme was used for binary fading interference channels in~\cite[Appendix A]{AlirezaDelayed}.

\section{Outer Bounds on DoF of Time-Varying AF} \label{converse}

Consider a time-varying AF $(n,R_1,R_2)$-scheme $\mathcal{C}$ with power constraint $P$, and error probability $P_e^n$ such that $P_e^n \rightarrow 0$ as $n \rightarrow \infty$. We will prove that $R_1+R_2 \leq (2/3) \log P + o (\log P)$.
Let $\mu_k$ and $\lambda_k$ denote the amplifying factors of $\mathcal{C}$ at time $k$ of relays $u$ and $v$, respectively.
Consider the end-to-end channel matrix $\mathbf{G}_k$ (defined in (\ref{Gk})) created by scheme $\mathcal{C}$ at time $k$. Note that the $i$-th column (row) of $\mathbf{G}_k$ ($i \in \{1,2\}$) corresponds to a linear combination of columns of $\mathbf{H_1}$ ($\mathbf{H_2}$) with coefficients $\mu_k h_{s_i,u}$ and $\lambda_k h_{s_i,v}$ ($\mu_k h_{u,d_i}$ and $\lambda_k h_{v,d_i}$). Also, the entries of the main diagonal are linear combinations of the columns of $\mathbf{H^1}$ (defined in (\ref{eq:conditions})) with coefficients $\mu_k$ and $\lambda_k$; similarly, the entries of the counterdiagonal are linear combinations of the columns of $\mathbf{H^2}$ (defined in (\ref{eq:conditions})) with coefficients $\mu_k$ and $\lambda_k$. Since by conditions (c-1), (c-2), and (c-3), specified in (\ref{eq:conditions}), all channel coefficients are non-zero, and $\mathbf{H_1}$, $\mathbf{H_2}$, $\mathbf{H^1}$, and $\mathbf{H^2}$ have full rank, it follows that no pair of entries in $\mathbf{G}_k$ can be zero unless $\lambda_k=\mu_k=0$. Therefore, at each time $k$ either $\mathbf{G}_k$ has at most one zero entry or $\mathbf{G}_k=\mathbf{0}$. As a result, if $\mathbf{G}_k$ is non-zero at any time $k$, then it belongs to one of the states shown in Figure \ref{states}. Asterisks denote non-zero entries. We denote the collective state $(C_1,C_2,C_3)$  by $C$. %\vspace{2mm}

\begin{figure}[h]
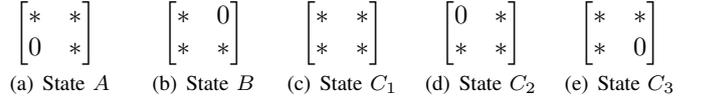

\centering
\hspace{-0.4in}
\subfigcapskip = 4pt
\subfigure[loose][State $A$]{$\begin{array}{lr}  &  \\ &   \end{array} \begin{bmatrix} *  & * \\ 0 & * \end{bmatrix} \begin{array}{lr} &  \\ &  \end{array}  $ } \hspace{-0.4in} 
\subfigure[loose][\vspace{15pt} \newline State $B$]{$\begin{array}{lr}  &  \\  &  \end{array} \begin{bmatrix} *  & 0 \\ * & * \end{bmatrix} \begin{array}{lr}  &  \\  &  \end{array}$} \hspace{-0.4in}
\subfigure[loose][\vspace{15pt} \newline State $C_1$]{$\begin{array}{lr}  &  \\  &  \end{array} \begin{bmatrix} *  & * \\ * & * \end{bmatrix} \begin{array}{lr}  &  \\  &  \end{array}$} \hspace{-0.4in}
\subfigure[loose][\vspace{15pt} \newline State $C_2$]{$\begin{array}{lr}  &  \\  &  \end{array} \begin{bmatrix} 0  & * \\ * & * \end{bmatrix} \begin{array}{lr}  &  \\  &  \end{array}$} \hspace{-0.4in}
\subfigure[loose][\vspace{15pt} \newline State $C_3$]{$\begin{array}{lr}  &  \\  &  \end{array} \begin{bmatrix} *  & * \\ * & 0 \end{bmatrix} \begin{array}{lr}  &  \\  &  \end{array}$} \hspace{-0.4in}
\caption{If $\max\{|\mu_k|, |\lambda_k| \} > 0$ at a time $k$,  then the end-to-end channel matrix $\mathbf{G}_k$ is in one of the above states. Asterisks denote non-zero entries.}\label{states}
\end{figure}

Similarly to equation (\ref{RSignal}), we will write the vector of $n$  received signals at destination $d_i$ (with an abuse of notation)
\begin{equation} \label{nRS}
Y_i^n=\alpha_i^n X_1^n + \beta_i^n X_2^n + Z_i^n, \text{\quad $i \in \{1,2\}$, }
\end{equation}
where $\alpha_i^n$ and $\beta_i^n$ are understood as $n\times n$ diagonal matrices, where the $j^{th}$ entries of the diagonals are respectively $\alpha_{i,j}$, and $\beta_{i,j}$ ($i\in\{1,2\}, j\in\{1,\dots,n\}$). Similarly, for any $L \subset \{1,2,\dots,n\}$, we write
\begin{equation}
Y_i^L=\alpha_i^L X_1^L+\beta_i^L X_2^L+Z_i^L,
\end{equation}
where $\alpha_i^L$ and $\beta_i^L$ are $|L| \times |L|$ diagonal matrices, whose diagonal entries are respectively $\{ {\alpha_{i,l}\}}_{l \in L}$ and $\{ {\beta_{i,l}\}}_{l \in L}$ ($i \in \{1,2\})$.

Now, for any code $\mathcal{C}$ (with AF coefficients $\lambda_k$ and $\mu_k$, $k=1,\dots,n$), we define the set $A_{\mathcal{C}}$  as 
\begin{equation*} \label{defA}
A_{\mathcal{C}}=\{k \in \{1,2,\dots,n\}: \mu_k h_{u,d_2} h_{s_1,u} + \lambda_k h_{v,d_2} h_{s_1,v}=0\}.
\end{equation*}
Similarly, let $B_{\mathcal{C}}$, $C_{\mathcal{C}}$, $C_{1,\mathcal{C}}$, $C_{2,\mathcal{C}}$, and $C_{3,\mathcal{C}}$ be the sets of time indices corresponding to states $B$, $C$, $C_1$, $C_2$, and $C_3$, respectively. Also, let $S_{\mathcal{C}}=A_{\mathcal{C}} \cup  B_{\mathcal{C}} \cup C_{\mathcal{C}}$. So that $S^c_{\mathcal{C}}=\{ k \in \{1,2,\dots,n\}: \mu_k=\lambda_k=0\}$. Note that the previously defined sets are deterministic and well-defined, since the channel gains are fixed  and we are considering a specific scheme $\mathcal{C}$ which fixes $\mu_k$ and $\lambda_k$ for all $k$. Also, for ease of notation, we will drop the subscript $\mathcal{C}$ in the rest of this section and refer to those sets as $A$, $B$, $C$, $C_1$, $C_2$, $C_3$, and $S$. We now state our main lemma which yields $\mathcal{D} \leq 4/3$.
\begin{Lemma} \label{bounds} For any  time-varying AF $(n,R_1,R_2)$-scheme, $\mathcal{C}$, with power constraint $P$ and associated sets $A$, $B$, $C$, as defined above, we have  %{\bf{(Main Lemma)}}
\begin{align*}
R_1+R_2 & \leq \frac{1}{2} \left(1+ \frac{|C|}{n} \right) \log_2 P+ \tau_1,  \tag{Bound 1} \\
R_1+R_2 & \leq \frac{1}{2} \left(1+ \frac{|B|}{n} \right) \log_2 P + \tau_2,  \tag{Bound 2} \\
R_1+R_2 & \leq \frac{1}{2} \left(1+ \frac{|A|}{n} \right) \log_2 P + \tau_3,  \tag{Bound 3}
\end{align*}
where $\tau_1$, $\tau_2$, and $\tau_3$ are constants that do not depend on $P$.
\end{Lemma} 
\noindent Before proving Lemma \ref{bounds}, we first demonstrate how it yields $\mathcal{D} \leq 4/3$. Suppose that the Lemma is true. Then, by taking the minimum of the three bounds, we get 
\begin{align} \label{conv}
R_1+R_2  & \leq \underset{L \in \{A,B,C\}}{\min} \frac{1}{2} \left( \! 1 \! + \! \frac{|L|}{n} \right)\log_2 \! P + \tau  \leq \frac{2}{3} \log_2 \! P  \!+\! \tau, \notag \\
\Rightarrow \mathcal{D} & \leq 4/3,
\end{align}
where $\tau=\max(\tau_1,\tau_2,\tau_3)$, and the second inequality follows from the fact that $\min \{|A|,|B|,|C|\} \leq n/3$ since $|A|+|B|+|C| \leq n$. We will now go back to proving the bounds in Lemma \ref{bounds}.\\

\noindent {\bf{Proof of Bound (1) in Lemma \ref{bounds}}}\\
Recall  that $S=A \cup B \cup C$. Then it is easy to show that $I(X_1^n;Y_1^n)=I(X_1^S;Y_1^S)$. 
 Similarly, $I(X_2^n; Y_2^n) = I(X_2^{S}; Y_2^{S})$.  Now, using Fano's inequality, we get
\begin{align} \label{Fano1}
& n\left(R_1+R_2\right) \leq I\left(X_1^n; Y_1^n\right)+ I\left(X_2^n; Y_2^n\right) + n{\epsilon}_n \notag \\
& ~~ = I\left(X_1^S; Y_1^S\right)+ I\left(X_2^S; Y_2^S\right)+ n{\epsilon}_n  \notag \\
& \hspace{-4.5mm} \stackrel{(S=A \cup B \cup C)}{=} I\left(X_1^S; Y_1^A, Y_1^B\right) + I\left(X_2^S; Y_2^A, Y_2^B\right) + n{\epsilon}_n \notag \\ 
& \quad ~~ +  I\left(X_1^S; Y_1^C | Y_1^A, Y_1^B\right) + I\left(X_2^S; Y_2^C | Y_2^A, Y_2^B\right),   
\end{align}
where $\epsilon_n \rightarrow 0$, as $P_e^n \rightarrow 0$. Now, we bound the last two terms:
\begin{align} \label{CC}
I(X_i^S; Y_i^C | Y_i^A, Y_i^B) & \leq h(Y_i^C) - h(Y_i^C | Y_i^A, Y_i^B, X_i^S, X_{\bar{i}}^C) \notag \\
& = h(Y_i^C) - h(Z_i^C),
\end{align}
where $i \in \{1,2\} $, $\bar{i}= 3-i$, and the equality follows from the fact that noise is independent of $\{W_1,W_2\}$ and of noise terms at other time steps. Now, to bound the first two terms in (\ref{Fano1}), consider the following chain of inequalities.
\begin{align} \label{AB1} 
&  I(X_1^S; Y_1^A, Y_1^B) + I(X_2^S; Y_2^A, Y_2^B) \notag \\
& \leq h( Y_1^A) + h(Y_1^B) - h( Y_1^A, Y_1^B | X_1^S)  \notag \\
& \quad +  h( Y_2^A) + h(Y_2^B) - h( Y_2^A, Y_2^B | X_2^S)  \notag \\
& = h( Y_1^A) + h(Y_1^B) - h(\beta_1^AX_2^A+Z_1^A, Z_1^B | X_1^S) \notag  \\
& \quad + h( Y_2^A) + h(Y_2^B) - h(\alpha_2^BX_1^B+Z_2^B, Z_2^A | X_2^S)  \notag \\
& \stackrel{\text{(a)}}= h(Y_1^A)+h(Y_2^B)-h(Z_1^B) -h(Z_2^A)  \notag \\ 
& \quad +\left[ h(\alpha_1^BX_1^B+Z_1^B)- h(\alpha_2^BX_1^B+Z_2^B) \right]  \notag \\
& \quad + \left[h(\beta_2^AX_2^A+Z_2^A) - h(\beta_1^AX_2^A+Z_1^A)\right], 
\end{align}
where (a) follows from the fact that $W_1$ and $W_2$ are independent, noise and $(W_1,W_2)$ are independent, and noise terms at different time steps are independent. Now, consider the following lemma.
\begin{Lemma} \label{side}
Let $X, Y, Z$ be two  random vectors of size $n$, such that $X \ci(Y,Z)$. Let $\mathbf{M}$ and $\mathbf{M'}$ be two $n \times n$ constant invertible matrices. Then
\begin{equation*}
\begin{split}
&  h(\mathbf{M} X + Y) -  h(\mathbf{M'} X + Z) \leq \\
 & h(\mathbf{M'}\mathbf{M}^{-1}Y  - Z) - h (Z |Y) - \log \left | \det\left(\mathbf{M'}\mathbf{M}^{-1} \right) \right| .
\end{split}
\end{equation*} 
\end{Lemma} 
\begin{proof}
\begin{equation*}
\begin{split}
&  h(\mathbf{M} X + Y) -  h(\mathbf{M'} X + Z) \\
& = h(\mathbf{M'}X + \mathbf{M'}\mathbf{M}^{-1} Y) - h(\mathbf{M'} X + Z)  \\ 
& \quad - \log \left | \det\left(\mathbf{M'}\mathbf{M}^{-1} \right) \right| \\
& \leq h(\mathbf{M'}X + \mathbf{M'}\mathbf{M}^{-1} Y) - h(\mathbf{M'} X + Z | \mathbf{M'}\mathbf{M}^{-1} Y - Z) \\
& \quad  - \log \left | \det\left(\mathbf{M'}\mathbf{M}^{-1} \right) \right| \\
& = - h(\mathbf{M'}X + \mathbf{M'}\mathbf{M}^{-1} Y | \mathbf{M'}\mathbf{M}^{-1} Y - Z) \\ 
& \quad + h(\mathbf{M'}X + \mathbf{M'}\mathbf{M}^{-1} Y) - \log \left | \det\left(\mathbf{M'}\mathbf{M}^{-1} \right) \right| \\
& = I(\mathbf{M'}X + \mathbf{M'}\mathbf{M}^{-1} Y; \mathbf{M'}\mathbf{M}^{-1} Y - Z) \\
& \quad  - \log \left | \det\left(\mathbf{M'}\mathbf{M}^{-1} \right) \right| \\
& = h(\mathbf{M'}\mathbf{M}^{-1} Y - Z) - h(\mathbf{M'}\mathbf{M}^{-1} Y - Z | \mathbf{M'}X + \mathbf{M'}\mathbf{M}^{-1} Y) \\
& \quad - \log \left | \det\left(\mathbf{M'}\mathbf{M}^{-1} \right) \right| \\
& \leq h(\mathbf{M'}\mathbf{M}^{-1} Y - Z) - h(\mathbf{M'}\mathbf{M}^{-1} Y - Z | X,Y) \\
& \quad  - \log \left | \det\left(\mathbf{M'}\mathbf{M}^{-1} \right) \right| \\
 & \leq h(\mathbf{M'}\mathbf{M}^{-1}Y  - Z) - h (Z |Y) - \log \left | \det\left(\mathbf{M'}\mathbf{M}^{-1} \right) \right| .
\end{split} 
\end{equation*} 
\end{proof}
\noindent Then we can apply Lemma \ref{side} on the bracketed terms in equation (\ref{AB1}), where for the first term $\{X=X_1^B$, $Y=Z_1^B$, $Z=Z_2^B$, $\mathbf{M}=\alpha_1^B$, and $\mathbf{M'}=\alpha_2^B\}$, and for the second term $\{X=X_2^A$, $Y=Z_2^A$, $Z=Z_1^A$, $\mathbf{M}=\beta_2^A,$ and $\mathbf{M'}=\beta_1^A\}$.  So by setting $\mathbf{M_1}=(\alpha_2^B)(\alpha_1^B)^{-1}$, $\mathbf{M_2}=(\beta_1^A)(\beta_2^A)^{-1}$, we get
\begin{align} \label{AB2}
&  I(X_1^S; Y_1^A, Y_1^B) + I(X_2^S; Y_2^A, Y_2^B) \notag \\
&  \leq h(Y_1^A)+ h(Y_2^B)-h(Z_1^B) -h(Z_2^A)  \notag \\
& \quad + h\left(\mathbf{M_1}Z_1^B-Z_2^B\right) - h\left(Z_2^B |Z_1^B\right)-\log \left|\det\left(\mathbf{M_1}\right)\right|   \notag \\ 
&  \quad + h\left(\mathbf{M_2}Z_2^A-Z_1^A\right) - h\left(Z_1^A | Z_2^A\right)  - \log \left|\det\left(\mathbf{M_2}\right) \right|   \notag \\
&  \leq h(Y_1^A)+ h(Y_2^B)  +\gamma_1 n,
\end{align}
\noindent where $\gamma_1$ is a constant that does not depend on $P$. Now, by equations (\ref{Fano1}), (\ref{CC}), and (\ref{AB2}), we get
\begin{equation} \label{bABCC}
n(R_1+R_2) \leq h(Y_1^A)+h(Y_2^B)+h(Y_1^C)+h(Y_2^C) + \gamma_2 n,
\end{equation}
where $\gamma_2$ is a constant that does not depend on $P$. Now, we bound $h(Y_1^A)$ by
\begin{equation*}
\begin{split}
& h(Y_1^A) - |A|\log_2 (2 \pi e)/2 \leq \sum_{k \in A} h(Y_{1,k}) -|A|\log_2 (2 \pi e)/2 \\
& \stackrel{\text{(a)}}\leq \sum_{k \in A}  \frac{1}{2} \log_2 ( \alpha_{1,k}^2 \E[X_{1,k-1}^2]+\beta_{1,k}^2 \E[X_{2,k-1}^2]+ \E[Z_{1,k}^2]),
\end{split}
\end{equation*} 
where (a) is true because Gaussian distribution maximizes differential entropy. Define $M_{i,j}$ ($i,j \in \{1,2\}$) and $M$ as
\begin{equation} \label{var1}
\begin{split} 
& M_{i,j} = \underset{\mu \in \mathcal{U}, \lambda \in \mathcal{V}}{\max} (\mu h_{u,d_i}h_{s_j,v}+\lambda h_{v,d_i}h_{s_j,v})^2, \\
& M = \underset{j \in \{1,2\}}{\max} ~ \underset{i \in \{1,2\}}{\max } (M_{i,j}) .
\end{split}
\end{equation}
Recall $\alpha_{1,k} = \mu_k h_{u,d_1}h_{s_1,u}+\lambda_k h_{v,d_1}h_{s_1,v}$. Then $\alpha_{1,k}^2 \leq M$, $\forall k$ . Similarly, $\beta_{1,k}^2 \leq M$, $\forall k$. Also, define $N$
\begin{equation} \label{var2}
N = \underset{i \in \{1,2\}}{\max } \left( \underset{\mu \in \mathcal{U}, \lambda \in \mathcal{V}}{\max} (h_{u,d_i}^2\mu^2+h_{v,d_i}^2\lambda^2) \right) + 1.
\end{equation}
\noindent Then $\E[Z_{1,k}^2]=((h_{u,d_1}\mu_k)^2+(h_{v,d_1}\lambda_k)^2+1) \leq N$, $\forall k$.
Thus
\begin{align} \label{bA1}
& h(Y_1^A) -|A|\log_2 (2 \pi e)/2 \notag  \\
& \leq  \sum_{k \in A} \frac{1}{2} \log_2 \left(M \E[X_{1,k-1}^2]+M \E[X_{2,k-1}^2]+ N\right) \notag \\
& \stackrel{\text{(a)}} \leq \frac{|A|}{2} \log_2 \left(N + M \frac{\sum_{k \in A} (\E[X_{1,k-1}^2]+\E[X_{2,k-1}^2])}{|A|} \right) \notag \\
& \stackrel{\text{(b)}} \leq \frac{|A|}{2} \log_2 \left(N+ 2MnP/|A| \right) \notag \\
& \leq \frac{|A|}{2} \log_2 P + \frac{|A|}{2} \log_2 (N+2Mn/|A|) \notag \\
& \leq \frac{|A|}{2} \log_2 P + \frac{|A|}{2} \log_2 N + \frac{1}{2} \log_2 \left(1+\frac{2Mn/N}{|A|} \right)^{|A|} \notag \\
& \stackrel{\text{(c)}} \leq \frac{A}{2} \log_2 P + \frac{n}{2} \log_2 (1+2M/N) + \frac{|A|}{2} \log_2 N,
\end{align}
where (a) follows from Jensen's inequality, (b) follows from the power constraint $P$, and (c) follows from the fact that the sequence $(1+x/m)^m$ is monotonically increasing in $m$ when $x>0$. 
Therefore, we can rewrite equation (\ref{bA1}) as
\begin{equation} \label{bA}
h(Y_1^A) \leq \frac{|A|}{2} \log_2 P + \gamma_3 n,
\end{equation}
where $\gamma_3$ is a constant that does not depend on $P$.
Similarly
\begin{align}
& h(Y_2^B) \leq  \frac{|B|}{2} \log_2 P + \gamma_4 n, \label{bB}\\
& h(Y_i^C) \leq  \frac{|C|}{2} \log_2 P + \gamma_{5,i} n, \text{\qquad $i \in \{1,2\}$}. \label{bC} 
\end{align} 
where $\gamma_4, \gamma_{5,1}$, and $\gamma_{5,2}$ are constants that do not depend on $P$. 
So, from equations (\ref{bABCC}), (\ref{bA}), (\ref{bB}), and (\ref{bC}) we get
\begin{equation*}
\begin{split}
n(R_1+R_2) & \leq \frac{1}{2}(|S|+|C|)\log_2 P + \tau_1 n \\
& \leq \frac{n}{2}(1+\frac{|C|}{n}) \log_2 P + \tau_1 n,
\end{split}
\end{equation*}
where $\tau_1$ is a constant that does not depend on $P$. \qquad \quad $\blacksquare$\\\\
\noindent {\bf{Proof of Bound (2) in Lemma \ref{bounds}}}\\
Define the set $E = C_1 \cup C_2$, and consider the following.
\begin{align}  \label{prebound2}
& n(R_1 + R_2) - n\epsilon_n \notag  \stackrel{\text{(a)}} \leq I(X_1^n; Y_1^n) + I (X_2^n; Y_2^n) \notag \\
&  = I(X_1^S; Y_1^S) + I (X_2^S; Y_2^S) \notag \stackrel{\text{(b)}}\leq I(X_1^S; Y_1^S) + I (X_2^S; Y_2^S| X_1^S) \notag  \\
& \leq h(Y_1^S) - h(\beta_1^AX_2^A+Z_1^A,Z_1^B,\beta_1^C X_2^C + Z_1^C) - h(Z_2^S) \notag  \\
& \quad + h(Y_2^B)+h(\beta_2^AX_2^A+Z_2^A,\beta_2^{E} X_2^{E}+Z_2^{E},Z_2^{C_3} ) \notag    \\
& \begin{array}{c@{\!\!\!}l}
&  \leq  [h(\beta_2^AX_2^A+Z_2^A,\beta_2^{E} X_2^{E}+Z_2^{E})  \notag \vspace{1mm} \\ 
& \:\: -h(\beta_1^AX_2^A+Z_1^A,\beta_1^{E} X_2^{E}+Z_1^{E})]  \notag \\
\end{array}
\begin{array}[c]{@{}l@{\,}l}
   \left. \begin{array}{c} \vphantom{0} 
   \\ \vphantom{0} \end{array} \right\} & (T1) \\
\end{array}\\
 & \: - h(Z_1^B, \beta_1^{C_3} X_2^{C_3} + Z_1^{C_3} | \beta_1^AX_2^A+Z_1^A,\beta_1^{E} X_2^{E} +Z_1^{E})  ~ (T2) \notag \\ 
& \: - h(Z_2^S) + h(Y_1^S) + h(Y_2^B)+ h(Z_2^{C_3}),  
\end{align}
where (a) follows from Fano's inequality, and (b) follows from the independence of $W_1$ and $W_2$. Now, we will bound the term $(T1)$. First, set $\mathbf{M_2} = (\beta_1^A) (\beta_2^A)^{-1}$, and $\mathbf{M_3} = (\beta_1^E)(\beta_2^E)^{-1}$. Then note
\begin{align} \label{temp}
& h(\beta_2^AX_2^A+Z_2^A,\beta_2^E X_2^E+Z_2^E) \notag \\
& - h(\beta_1^AX_2^A+Z_1^A,\beta_1^E X_2^E+Z_1^E) \leq \notag \\
& h(\mathbf{M_2} Z_2^A - Z_1^A, \mathbf{M_3} Z_2^E - Z_1^E) - h(Z_1^A, Z_1^E | Z_2^A, Z_2^E) \notag \\
&- \log \left| \det(\mathbf{M_2})\det(\mathbf{M_3})\right|,
\end{align}
where the inequality follows from Lemma \ref{side} and the fact that \\ $h(\mathbf{M} X, Y) = h(X,Y) + \log |\det\mathbf{M} |$. Now, we bound $(T2)$:
\begin{align} \label{temp2}
& h(Z_1^B, \beta_1^{C_3} X_2^{C_3} + Z_1^{C_3} | \beta_1^A X_2^A + Z_1^A, \beta_1^E X_2^E + Z_1^E) \geq \notag \\
& h(Z_1^B,Z_1^{C_3}|X_2^{C_3},\beta_1^A X_2^A + Z_1^A, \beta_1^E X_2^E + Z_1^E)  \geq \notag \\
& h(Z_1^B) + h(Z_1^{C_3}).
\end{align} 
Then, by equations (\ref{prebound2}), (\ref{temp}), and (\ref{temp2}), we get
\begin{align} \label{bound2}
& n(R_1 + R_2)   \leq h(Y_1^S) + h(Y_2^B) -h(Z_1^B) - h(Z_1^{C_3}) \notag \\
&  \quad + h(\mathbf{M_2} Z_2^A - Z_1^A, \mathbf{M_3} Z_2^E - Z_1^E) -h(Z_2^S) -h(Z_2^{C_3})   \notag \\
& \quad - h(Z_1^A, Z_1^E | Z_2^A, Z_2^E) - \log \left| \det(\mathbf{M_2})\det(\mathbf{M_3})\right|  \notag \\
& \leq h(Y_1^S) + h(Y_2^B) + \gamma_6 n,
\end{align}
where $\gamma_6$ is a constant that does not depend on $P$. Now,
similarly to (\ref{bA}), we bound $h(Y_1^S)$ as 
\begin{equation} \label{b2n}
h(Y_1^S) \leq \frac{|S|}{2}\log_2 P + \gamma_7 n, 
\end{equation}
where $\gamma_7$ is a constant that does not depend on $P$.
Then, from equations (\ref{bound2}), (\ref{b2n}), and (\ref{bB}), we get
\begin{equation*}
\begin{split}
n(R_1+R_2) & \leq \frac{1}{2}(|S|+|B|)\log_2 P + \tau_2 n, \\
& \leq \frac{n}{2}(1+\frac{|B|}{n}) \log_2 P + \tau_2 n,
\end{split}
\end{equation*}
where $\tau_2$ is a constant that does not depend on $P$. \qquad $\blacksquare$ \\
The proof of the third bound is similar, and thus omitted.

\section{Concluding Remarks} \label{conclusion}
In this paper, we analyzed the sum-DoF of the two-hop IC with real constant coefficients when relays are restricted to perform time-varying AF schemes. We showed that 4/3 sum-DoF is achievable using such schemes. In~\cite{JournalPrep}, we show that 4/3 is an upper bound for \emph{vector} linear schemes as well.  Although we considered \emph{real} channel gains, the ideas of this paper can be extended to complex channels, for which it was previously shown in~\cite{Jafar} that 3/2 sum-DoF can be achieved using linear schemes. We show in~\cite{JournalPrep} that by utilizing time-varying AF schemes, a sum-DoF of 5/3 can be achieved. We also extend the scheme for MIMO channels with real channel gains to achieve a sum-DoF of $2M-2/3$, where $M$ is the number of antennas at each node.
Future research may consider the impact of time-varying AF strategies in more general two-unicast networks, such as the layered networks studied in~\cite{Ilan2Unicast}.

\section*{Acknowledgement}
The research of I. Issa, S. L. Fong, and A. S. Avestimehr is supported in part by NSF Grants CAREER 0953117, CCF-1161720, Samsung Advanced Institute of Technology (SAIT), and AFOSR YIP award.

\bibliographystyle{IEEEtran}
\bibliography{IEEEabrv,database}

\end{document}